\documentclass[journal]{IEEEtran}

\ifCLASSINFOpdf
\else
   \usepackage[dvips]{graphicx}
\fi
\usepackage{url}
\usepackage{graphicx}
\graphicspath{{figure/}}
\usepackage{subcaption}
\usepackage{amssymb}
\usepackage{amsmath}
\usepackage{amsthm}
\usepackage{algorithm}
\usepackage{algorithmic}

\usepackage{cite}
\usepackage{orcidlink}

\usepackage{threeparttable}
\usepackage{booktabs}
\usepackage{caption}
\usepackage{hyperref}

\newtheorem{theorem}{Theorem}

\begin{document}

\title{A Marginal Distributionally Robust Kalman Filter for Sensor Fusion}

\author{Weizhi Chen\orcidlink{0000-0003-0066-4053}, Yaowen Li\orcidlink{0000-0002-3050-6390}, Yu Liu\orcidlink{0000-0002-5216-3181}, and You He\orcidlink{0000-0002-6111-340X}
\thanks{This work was supported by National Natural Science Foundation of China under Grant {62388102}, {62425117}, and {62401336}. (\textsl{Corresponding author: Yaowen Li}.)}
\thanks{Weizhi Chen and Yaowen Li are with the Shenzhen International Graduate School, Tsinghua University, Shenzhen 518055, China (e-mail: cwz22@mails.tsinghua.edu.cn; liyw23@sz.tsinghua.edu.cn.}%
\thanks{Yu Liu and You He are with the Department of Electronic Engineering, Tsinghua University, Beijing 100084, China (e-mail: liuyu77360132@126.com; heyou@mail.tsinghua.edu.cn).}
}

\markboth{Journal of \LaTeX\ Class Files, Vol. 14, No. 8, August 2015}
{Shell \MakeLowercase{\textit{et al.}}: Bare Demo of IEEEtran.cls for IEEE Journals}
\maketitle

\begin{abstract}

This paper proposes a moment-constrained marginal distributionally robust Kalman filter (MC-MDRKF) for centralized state estimation in multi-sensor systems with unknown sensor noise correlations. We first derive a robust static estimator and then extend it to dynamic systems for the MC-MDRKF algorithm. The static estimator defines a marginal distributional uncertainty set using moment constraints and formulates a minimax optimization problem to robustly address unknown correlations. We prove that this minimax problem admits an equivalent convex optimization formulation, enabling efficient numerical solutions. The resulting MC-MDRKF algorithm recursively updates state estimates in dynamic state-space models. Simulation results demonstrate the superiority and robustness of the proposed method in a multi-sensor target tracking scenario.

\end{abstract}


\IEEEpeerreviewmaketitle

\section{Introduction}

\IEEEPARstart{S}{tate} estimation for multi-sensor information fusion is vital for applications including target tracking\cite{8861414}, power systems\cite{8805459}, and control automation\cite{9479686}. Classical Bayesian theory uses the centralized Kalman filter for optimal multi-sensor fusion, assuming known system dynamics and noise statistics\cite{willner1976kalman}. Many variations have been developed to enhance robustness and accuracy in complex environments\cite{liu2017robust, chen2015networked}.

However, the classical Bayesian optimal estimator assumes that measurements from multiple sensors are either independent or have known correlations. The problem of marginal distributional uncertainty arises when these correlations are unknown or difficult to estimate in real time, causing significant challenges to centralized sensor fusion.

A promising approach to handling marginal distributional uncertainty is distributionally robust optimization (DRO), where decisions are optimized against worst-case distributions within a predefined uncertainty set. For instance, Fan et al.\cite{fan2024distributionally} investigate distributionally robust optimization with marginal and copula ambiguity for portfolio optimization, employing the Wasserstein distance. Building on this work, Niu et al.\cite{niu2023marginal} introduce a marginal distributionally robust MMSE estimation for multi-sensor systems using Kullback-Leibler (KL) divergence constraints, demonstrating the effectiveness of marginal uncertainty sets in handling distributional uncertainty in centralized multi-sensor fusion systems. Notably, both works focus solely on static estimation.

However, KL divergence-based methods have been shown to face challenges when extended to dynamic state estimation in state space, even in single-sensor scenarios~\cite{zorzi2016robust, zorzi2019distributed}. Specifically, they may lack fine-grained robustness and can be computationally intensive, limiting their effectiveness in handling dynamic uncertainties~\cite{shafieezadeh2018wasserstein}. In contrast, for Gaussian distributions, moment-constrained methods can be computationally efficient and more effectively handle distributional uncertainties~\cite{Rahimian2019DistributionallyRO}. Nevertheless, existing moment-constrained methods have predominantly been applied to single-sensor scenarios~\cite{wang2021robust, wang2021distributionally}, and in multi-sensor cases, the unknown sensor noise correlation becomes a significant challenge and entails establishing an specific marginal distributional uncertainty set and solving a completely new optimization problem.

In this paper, we propose the moment-constrained marginal distributionally robust Kalman filter (MC-MDRKF) to address marginal distributional uncertainty in multi-sensor systems.

The contributions are as follows: (1) A moment-constrained marginal distributional uncertainty set is devised for multi-sensor fusion to characterize unknown sensor noise correlation; (2) A robust static state estimator is developed by formulating and solving a minimax optimization problem over this uncertainty set, which is further shown to be equivalently reformulated as a convex optimization problem for efficient computation; (3) The MC-MDRKF algorithm is developed for robust centralized state estimation in multi-sensor systems by extending the static estimator to the state space model.


\section{Static Estimation Under Moment-constrained Marginal Distributional Uncertainty}

\subsection{System Model}

The system includes a fusion center and $p$ sensors, each providing an observation vector $\boldsymbol{y}^i \in \mathbb{R}^{m_i} (i=1,\dots,p)$ related to the random state vector $\boldsymbol{x}$. The fusion center uses the combined observations $\boldsymbol{y} = [(\boldsymbol{y}^1)^\top, \dots, (\boldsymbol{y}^p)^\top]^\top \in \mathbb{R}^{m}$, where $m=\Sigma_{i=1}^p m_i$, to estimate $\boldsymbol{x}$ with an estimator $\psi:\mathbb{R}^m\rightarrow \mathbb{R}^n$ giving $\psi(\boldsymbol{y})$. The joint vector $\boldsymbol{z}=[\boldsymbol{x}^{\top},\boldsymbol{y}^{\top}]^{\top}$ in $\mathbb{R}^{n+m}$ has probability density $\mathbb{P}$. The estimator’s mean squared error (MSE) is defined as
\begin{equation}
J(\mathbb{P},\psi)=\mathbb{E}^{\mathbb{P}}\left[\left \| \boldsymbol{x}-\psi(\boldsymbol{y}) \right \|^2_2 \right]
\label{def1}
\end{equation}

The fusion center has access to the nominal marginal distributions $\mathbb{P}_i$ for each sensor  $i = 1, \dots, p$, where each $\mathbb{P}_i$ is assumed to be Gaussian: $\mathbb{P}_i = \mathcal{N}(\boldsymbol{\mu}_{\boldsymbol{x}, \boldsymbol{y}^i}, \boldsymbol{\Sigma}_{\boldsymbol{x}, \boldsymbol{y}^i})$.

\subsection{Moment-constrained Marginal Distributional Uncertainty Set}

Due to uncertainty in the sensor noise correlations, the joint distribution of $\boldsymbol{x}$ and $\boldsymbol{y}$ is not precisely known, and only nominal marginal distributions $\mathbb{P}_i$ are available. To account for this uncertainty, we define a moment-constrained distributional uncertainty set  $\mathcal{P}_{\mathrm{M}}$ , which consists of distributions that satisfy certain moment constraints on the marginals of  $\boldsymbol{x}$  and  $\boldsymbol{y}^i$ . Specifically,
\begin{equation}
\begin{aligned}
    & \mathcal{P}_{\mathrm{M}} := \{\mathbb{Q} \in \mathcal{F} : \forall i=1,\dots,p, \\
    & (\boldsymbol{c}_{\boldsymbol{x},\boldsymbol{y}^i} - \boldsymbol{\mu}_{\boldsymbol{x},\boldsymbol{y}^i})^T \boldsymbol{\Sigma}_{\boldsymbol{x},\boldsymbol{y}^i}^{-1}(\boldsymbol{c}_{\boldsymbol{x},\boldsymbol{y}^i} - \boldsymbol{\mu}_{\boldsymbol{x},\boldsymbol{y}^i}) \leq \gamma_{3,i}, \\
    & \boldsymbol{S}_{\boldsymbol{x},\boldsymbol{y}^i} + (\boldsymbol{c}_{\boldsymbol{x},\boldsymbol{y}^i} - \boldsymbol{\mu}_{\boldsymbol{x},\boldsymbol{y}^i})(\boldsymbol{c}_{\boldsymbol{x},\boldsymbol{y}^i} - \boldsymbol{\mu}_{\boldsymbol{x},\boldsymbol{y}^i})^T \preceq \gamma_{2,i} \boldsymbol{\Sigma}_{\boldsymbol{x},\boldsymbol{y}^i}, \\
    & \boldsymbol{S}_{\boldsymbol{x},\boldsymbol{y}^i} + (\boldsymbol{c}_{\boldsymbol{x},\boldsymbol{y}^i} - \boldsymbol{\mu}_{\boldsymbol{x},\boldsymbol{y}^i})(\boldsymbol{c}_{\boldsymbol{x},\boldsymbol{y}^i} - \boldsymbol{\mu}_{\boldsymbol{x},\boldsymbol{y}^i})^T \succeq \gamma_{1,i} \boldsymbol{\Sigma}_{\boldsymbol{x},\boldsymbol{y}^i}
\}
\end{aligned}
\label{def3}
\end{equation}
where $\mathcal{F}$ includes all distributions of \(\boldsymbol{z}\) with finite second-order moments. $\boldsymbol{c}_{\boldsymbol{x},\boldsymbol{y}^i}$ and $\boldsymbol{S}_{\boldsymbol{x},\boldsymbol{y}^i}$ are the mean and covariance of the potential distributions $\mathbb{Q}_{\boldsymbol{x}, \boldsymbol{y}^i}$, respectively. The non-negative constants $\gamma_{1,i}, \gamma_{2,i}, \gamma_{3,i}$ are adjustable.

\subsection{Minimax Optimization Problem for Robust State Estimation}

To robustly estimate the state $\boldsymbol{x}$ despite marginal distributional uncertainty, we design a robust static state estimator $\psi$ by solving the following minimax optimization problem:
\begin{equation}
    \inf_{\psi \in \mathcal{L}} \sup_{\mathbb{Q} \in \mathcal{P}_{\mathrm{M}}} \mathbb{E}^{\mathbb{Q}}\left[\|\boldsymbol{x} - \psi(\boldsymbol{y})\|^2_2\right]
\label{original minimax}
\end{equation}
where $ \mathcal{L}$  is the set of all measurable functions mapping from  $\mathbb{R}^m $ to  $\mathbb{R}^n$ . The problem seeks an estimator  $\psi$  that minimizes the MSE under the least favorable distribution in  $\mathcal{P}_{\mathrm{M}}$.

\subsection{Convex Reformulation and Solution}

To solve \eqref{original minimax}, the upper and lower bounds of the problem are established respectively, and proved equivalent and solvable. Specifically, the upper bound can be reformulated as a convex problem, solvable as demonstrated in Theorem \ref{theorem1}. Theorem \ref{theorem2} establishes the lower bound and its equivalence to the upper bound, ensuring that solving the convex program effectively addresses the minimax problem. Theorem \ref{theorem3} further verifies that such solution is a saddle point of the original problem.

\begin{theorem}
\label{theorem1}
The upper bound of problem \eqref{original minimax} is derived by restricting $\mathcal{L}$ to affine estimators $\mathcal{A}$ as
\begin{equation}
\inf_{\psi \in \mathcal{A}} \sup_{\mathbb{Q} \in \mathcal{P}_\mathrm{M}} J(\mathbb{Q}, \psi)
\label{upper approx}
\end{equation}
This can be equivalently solved by
\begin{equation}
    \sup_{\boldsymbol{S} \in \mathcal{P}_\mathrm{M}'} \operatorname{Tr} \left( \boldsymbol{S}_{\boldsymbol{x}\boldsymbol{x}} - \boldsymbol{S}_{\boldsymbol{x}\boldsymbol{y}} \boldsymbol{S}_{\boldsymbol{y}\boldsymbol{y}}^{-1} \boldsymbol{S}_{\boldsymbol{y}\boldsymbol{x}} \right)
    \label{maxmin-II}
\end{equation}
where $\boldsymbol{S}$ meets the criteria of the moment-constrained set $\mathcal{P}_\mathrm{M}'$.
\end{theorem}

\begin{proof}
By restricting the set of estimators $\mathcal{L}$ to the set of affine estimators $\mathcal{A}$, defined by
\begin{equation}
    \mathcal{A} := \left\{ \psi \in \mathcal{L} \mid \exists \boldsymbol{A} \in \mathbb{R}^{n \times m}, \boldsymbol{b} \in \mathbb{R}^n, \text{s.t.} \psi(\boldsymbol{y}) = \boldsymbol{A}\boldsymbol{y} + \boldsymbol{b} \right\}
    \label{def6}
\end{equation}
an upper bound of the minimax problem \eqref{original minimax} can be formulated as in \eqref{upper approx}.

Inspired by the robust estimation approach outlined in \cite{wang2021robust}, the objective function of \eqref{upper approx} can be reformulated as follows:
\begin{equation}
\begin{aligned}
\inf _{\boldsymbol{A}, \boldsymbol{b}} \sup _{\boldsymbol{c},\boldsymbol{S}} & \left\langle\boldsymbol{I}, \boldsymbol{S}_{\boldsymbol{x} \boldsymbol{x}}+\boldsymbol{c}_{\boldsymbol{x}} \boldsymbol{c}_{\boldsymbol{x}}^{\top}\right\rangle+\left\langle\boldsymbol{A}^{\top} \boldsymbol{A}, \boldsymbol{S}_{\boldsymbol{y} \boldsymbol{y}}+\boldsymbol{c}_{\boldsymbol{y}} \boldsymbol{c}_{\boldsymbol{y}}^{\top}\right\rangle \\
& -\left\langle\boldsymbol{A}, \boldsymbol{S}_{\boldsymbol{x} \boldsymbol{y}}+\boldsymbol{c}_{\boldsymbol{x}} \boldsymbol{c}_{\boldsymbol{y}}^{\top}\right\rangle-\left\langle\boldsymbol{A}^{\top}, \boldsymbol{S}_{\boldsymbol{y} \boldsymbol{x}}+\boldsymbol{c}_{\boldsymbol{y}} \boldsymbol{c}_{\boldsymbol{x}}^{\top}\right\rangle \\
& +2\left\langle\boldsymbol{b}, \boldsymbol{A} \boldsymbol{c}_{\boldsymbol{y}}-\boldsymbol{c}_{\boldsymbol{x}}\right\rangle+\left\langle\boldsymbol{b}, \boldsymbol{b}\right\rangle
\end{aligned}
\label{upper approx-finite}
\end{equation}
where $\boldsymbol{c} = \mathbb{E}^{\mathbb{Q}}(\boldsymbol{z}) \in \mathbb{R}^{n+m}$ and $\boldsymbol{S} = \mathbb{E}^{\mathbb{Q}}(\boldsymbol{z}\boldsymbol{z}^\top) - \boldsymbol{c}\boldsymbol{c}^\top \in \mathbb{S}^{n+m}_+$, with $(\boldsymbol{c}, \boldsymbol{S})$ satisfying the constraints of the marginal uncertainty set:
\begin{equation}
\begin{aligned}
	& \mathcal{P}_\mathrm{M}':=\{ \boldsymbol{c}\in \mathbb{R}^{n+m}, \boldsymbol{S}\in \mathbb{S}^{n+m}: \forall i=1,\dots,p,  \\
	& \left(\boldsymbol{c}_{\boldsymbol{x},\boldsymbol{y}^i}-\boldsymbol{\mu}_{\boldsymbol{x},\boldsymbol{y}^i}\right)^{T} \boldsymbol{\Sigma}_{\boldsymbol{x},\boldsymbol{y}^i}^{-1}\left(\boldsymbol{c}_{\boldsymbol{x},\boldsymbol{y}^i}-\boldsymbol{\mu}_{\boldsymbol{x},\boldsymbol{y}^i}\right) \leq \gamma_{3,i},\\
	& \boldsymbol{S}_{\boldsymbol{x},\boldsymbol{y}^i}+\left(\boldsymbol{c}_{\boldsymbol{x},\boldsymbol{y}^i}-\boldsymbol{\mu}_{\boldsymbol{x},\boldsymbol{y}^i}\right)\left(\boldsymbol{c}_{\boldsymbol{x},\boldsymbol{y}^i}-\boldsymbol{\mu}_{\boldsymbol{x},\boldsymbol{y}^i}\right)^{T} \preceq \gamma_{2,i} \boldsymbol{\Sigma}_{\boldsymbol{x},\boldsymbol{y}^i}, \\
	& \boldsymbol{S}_{\boldsymbol{x},\boldsymbol{y}^i}+\left(\boldsymbol{c}_{\boldsymbol{x},\boldsymbol{y}^i}-\boldsymbol{\mu}_{\boldsymbol{x},\boldsymbol{y}^i}\right)\left(\boldsymbol{c}_{\boldsymbol{x},\boldsymbol{y}^i}-\boldsymbol{\mu}_{\boldsymbol{x},\boldsymbol{y}^i}\right)^{T} \succeq \gamma_{1,i} \boldsymbol{\Sigma}_{\boldsymbol{x},\boldsymbol{y}^i}\}
\end{aligned}
\label{def8}
\end{equation}
and $\langle \boldsymbol{A}^{\top}, \boldsymbol{B} \rangle := \operatorname{Tr}[\boldsymbol{A}^{\top} \boldsymbol{B}]$ denotes the trace inner product of two matrices $\boldsymbol{A}$ and $\boldsymbol{B}$.

Note \eqref{upper approx-finite} is constraint-free, quadratic and convex in terms of $\boldsymbol{b}$. Therefore, the optimal solution for $\boldsymbol{b}$ can be obtained by the first-order optimality condition:
\begin{equation}
	\boldsymbol{b}^\star = \boldsymbol{\mu}_{\boldsymbol{x}} - \boldsymbol{A}\boldsymbol{\mu}_{\boldsymbol{y}}
\label{sovb}
\end{equation}

This equality simplifies \eqref{upper approx-finite} to
\begin{equation}
	\inf _{\boldsymbol{A}} \sup _{\boldsymbol{S}}\left\langle\boldsymbol{I}, \boldsymbol{S}_{\boldsymbol{x} \boldsymbol{x}}\right\rangle+\left\langle\boldsymbol{A}^{\top} \boldsymbol{A}, \boldsymbol{S}_{\boldsymbol{y} \boldsymbol{y}}\right\rangle-\left\langle\boldsymbol{A}, \boldsymbol{S}_{\boldsymbol{x} \boldsymbol{y}}\right\rangle-\left\langle\boldsymbol{A}^{\top}, \boldsymbol{S}_{\boldsymbol{y} \boldsymbol{x}}\right\rangle
\label{obj1}
\end{equation}
which can be further written in a compact form as
\begin{equation}
	\inf _{\boldsymbol{A}} \sup _{\boldsymbol{S}}\left\langle\left[\begin{array}{cc}
\boldsymbol{I} & -\boldsymbol{A} \\
-\boldsymbol{A}^{\top} & \boldsymbol{A}^{\top} \boldsymbol{A}
\end{array}\right], \boldsymbol{S}\right\rangle
\label{obj2}
\end{equation}
which is subject to \eqref{def8}.

Note that the objective function \eqref{obj2} is independent of $\boldsymbol{c}$. Therefore, to maximize \eqref{obj2}, it is advantageous to have a larger feasible set for $\boldsymbol{S}$. This leads to the optimal solution for $\boldsymbol{c}$ being
\begin{equation}
	\boldsymbol{c}^\star = \boldsymbol{\mu}
\end{equation}

Since the uncertainty set \eqref{def8} is convex and compact in terms of $\boldsymbol{S}_k$ and the objective function in \eqref{obj2} is affine in $\boldsymbol{S}$ and positive-definite quadratic in $\boldsymbol{A}$, von Neumann’s min-max theorem\cite{fan1953minimax} holds, i.e.,
\begin{equation}
   	\inf _{\boldsymbol{A}} \sup _{\boldsymbol{S}}\left\langle\left[\begin{array}{cc}
\boldsymbol{I} & -\boldsymbol{A} \\
-\boldsymbol{A}^{\top} & \boldsymbol{A}^{\top} \boldsymbol{A}
\end{array}\right], \boldsymbol{S}\right\rangle
   = \sup _{\boldsymbol{S}} 	\inf _{\boldsymbol{}A}\left\langle\left[\begin{array}{cc}
\boldsymbol{I} & -\boldsymbol{A} \\
-\boldsymbol{A}^{\top} & \boldsymbol{A}^{\top} \boldsymbol{A}
\end{array}\right], \boldsymbol{S}\right\rangle
\label{eq3}
\end{equation}

Given that problem \eqref{obj2} for $\boldsymbol{A}$ is unconstrained, differentiable, and convex, the first-order optimality condition, i.e.,
\begin{equation}
	\boldsymbol{A}\boldsymbol{S}_{\boldsymbol{y}\boldsymbol{y}}-\boldsymbol{S}_{\boldsymbol{x}\boldsymbol{y}}=0
\end{equation}
gives the optimal solution of $\boldsymbol{A}$ as 
\begin{equation}
	\boldsymbol{A}^\star = \boldsymbol{S}_{\boldsymbol{x}\boldsymbol{y}}\boldsymbol{S}_{\boldsymbol{y}\boldsymbol{y}}^{-1}
\label{sova}
\end{equation}

With \eqref{sovb} and \eqref{sova}, \eqref{obj2} can be simplified to
\begin{equation}
	\sup_{\boldsymbol{S}\in \mathcal{P}_\mathrm{M}'} \operatorname{Tr}\left( \boldsymbol{S}_{\boldsymbol{x}\boldsymbol{x}}-\boldsymbol{S}_{\boldsymbol{x}\boldsymbol{y}} \boldsymbol{S}_{\boldsymbol{y}\boldsymbol{y}}^{-1}\boldsymbol{S}_{\boldsymbol{y}\boldsymbol{x}} \right)
\label{maxmin-II}
\end{equation}

This yields a convex semi-definite program, which can be solved numerically using semidefinite programming (SDP) solvers like SeDuMi via the CVX interface \cite{grant2014cvx}.

\end{proof}

\begin{theorem}
\label{theorem2}
The lower bound of the minimax problem \eqref{original minimax} can be established by reversing the order of the minimization and maximization operations, as follows:
\begin{equation}
    \sup_{\mathbb{Q} \in \mathcal{P}_\mathrm{M}} \inf_{\psi \in \mathcal{L}} J(\mathbb{Q}, \psi) = \sup_{\mathbb{Q} \in \mathcal{P}_\mathrm{M}} \inf_{\psi \in \mathcal{L}} \mathbb{E}^{\mathbb{Q}} \left[ \left\| \boldsymbol{x} - \psi(\boldsymbol{y}) \right\|^2 \right]
\end{equation}
Moreover, the equivalence of the upper and lower bounds is demonstrated as follows:
\begin{equation}
    \sup_{\mathbb{Q} \in \mathcal{P}_\mathrm{M}} \inf_{\psi \in \mathcal{L}} J(\mathbb{Q}, \psi) = \inf_{\psi \in \mathcal{L}} \sup_{\mathbb{Q} \in \mathcal{P}_\mathrm{M}} J(\mathbb{Q}, \psi) = \inf_{\psi \in \mathcal{A}} \sup_{\mathbb{Q} \in \mathcal{P}_\mathrm{M}} J(\mathbb{Q}, \psi).
\label{eq18}
\end{equation}
\end{theorem}

\begin{proof}

First, the following inequality is established:
\begin{equation}
	\sup_{\mathbb{Q}\in \mathcal{P}_\mathrm{M}} \inf_{\psi\in \mathcal{L}}  J(\mathbb{Q},\psi)\leq
		\inf_{\psi\in \mathcal{L}}  \sup_{\mathbb{Q}\in \mathcal{P}_\mathrm{M}} J(\mathbb{Q},\psi)\leq
		\inf_{\psi\in \mathcal{A}}  \sup_{\mathbb{Q}\in \mathcal{P}_\mathrm{M}} J(\mathbb{Q},\psi)
\label{eq19}
\end{equation}
where the first equality is due to the weak duality theorem\cite{boyd2004convex} and the second equality exploits the inclusion $\mathcal{A}\subseteq \mathcal{L}$.

Assuming $\boldsymbol{S}^\star$ is a solution of \eqref{maxmin-II}, it follows that
\begin{equation}
	\inf_{\psi\in \mathcal{A}}  \sup_{\mathbb{Q}\in \mathcal{P}_\mathrm{M}} J(\mathbb{Q},\psi) = 
	\operatorname{Tr}\left( \boldsymbol{S}_{\boldsymbol{x}\boldsymbol{x}}^{\star}-\boldsymbol{S}_{\boldsymbol{x}\boldsymbol{y}}^{\star} (\boldsymbol{S}_{\boldsymbol{y}\boldsymbol{y}}^{\star})^{-1}\boldsymbol{S}_{\boldsymbol{y}\boldsymbol{x}}^{\star}\right)
\label{eq20}
\end{equation}

Following \eqref{eq19} and \eqref{eq20}, we have
\begin{equation}
	\sup_{\mathbb{Q}\in \mathcal{P}_\mathrm{M}} \inf_{\psi\in \mathcal{L}}  J(\mathbb{Q},\psi)\leq
	\operatorname{Tr}\left( \boldsymbol{S}_{\boldsymbol{x}\boldsymbol{x}}^{\star}-\boldsymbol{S}_{\boldsymbol{x}\boldsymbol{y}}^{\star} (\boldsymbol{S}_{\boldsymbol{y}\boldsymbol{y}}^{\star})^{-1}\boldsymbol{S}_{\boldsymbol{y}\boldsymbol{x}}^{\star}\right)
\label{eq21}
\end{equation}

By the Bayesian estimation theory, for a Gaussian density
\begin{equation}
	\mathbb{Q}^{\star}_{\mathcal{N}}=\mathcal{N}(\boldsymbol{\mu},\boldsymbol{S}^\star) \in \mathcal{P}_\mathrm{M}
\label{eq17}
\end{equation}
the optimal value of the outer minimization problem of \eqref{original minimax} can be obtained as
\begin{equation}
	\inf_{\psi\in \mathcal{L}}  J(\mathbb{Q}^{\star}_{\mathcal{N}},\psi) = \operatorname{Tr}\left( \boldsymbol{S}_{\boldsymbol{x}\boldsymbol{x}}^{\star}-\boldsymbol{S}_{\boldsymbol{x}\boldsymbol{y}}^{\star} (\boldsymbol{S}_{\boldsymbol{y}\boldsymbol{y}}^{\star})^{-1}\boldsymbol{S}_{\boldsymbol{y}\boldsymbol{x}}^{\star}\right)
\label{eq22}
\end{equation}

Combining \eqref{eq22} with \eqref{eq19} and \eqref{eq20}, the proof is completed.

\end{proof}

Theorem \ref{theorem2} suggests solving the original problem \eqref{original minimax} is equivalent to solving either the upper or the lower bound problem. Furthermore, Theorem \ref{theorem3} established that the optimal solution pair $(\mathbb{Q}^\star_{\mathcal{N}}, \psi^\star)$ forms a saddle point for \eqref{original minimax}.

\begin{theorem}
\label{theorem3}
	Let $\mathcal{L}$ be the family of all measurable function from $\mathbb{R}^m$ to $\mathbb{R}^n$ and $\mathcal{P}_\mathrm{M}$ given by \eqref{def3} and $\psi^\star:\mathbb{R}^m\rightarrow \mathbb{R}^n$ be an affine function defined as 
	\begin{equation}
		\psi^\star(\boldsymbol{y})=\boldsymbol{A}^\star \boldsymbol{y} + \boldsymbol{b}^\star, \forall \boldsymbol{y} \in \mathbb{R}^m
	\label{eq23}
	\end{equation}
	where $(\boldsymbol{A},\boldsymbol{b})\in \mathbb{R}^m\times \mathbb{R}^n$ is given by \eqref{sova} and \eqref{sovb}. Then $(\mathbb{Q}^\star_{\mathcal{N}},\psi^\star)\in \mathcal{P}_\mathrm{M}\times \mathcal{L}$ is the saddle point solution of \eqref{original minimax}, i.e., $J(\mathbb{Q},\psi^\star)\leq J(\mathbb{Q}^\star_{\mathcal{N}},\psi^\star)\leq J(\mathbb{Q}^\star_{\mathcal{N}},\psi),\forall (\mathbb{Q},\psi)\in \mathcal{P}_{\mathrm{M}}\times \mathcal{L}$, where $J$ and $\mathbb{Q}^\star_{\mathcal{N}}$ are defined by  \eqref{original minimax} and \eqref{eq17}, respectively.
\end{theorem}

\begin{proof}

First, Theorem \ref{theorem2} already implies $\mathbb{Q}^\star_{\mathcal{N}}\in \mathcal{P}_\mathrm{M}$. Then $(\mathbb{Q}^\star_{\mathcal{N}},\psi^\star)\in \mathcal{P}_\mathrm{M}\times \mathcal{L}$ is a saddle point of $J(\mathbb{Q},\psi)$ if and only if
\begin{equation}
	\sup_{\mathbb{Q}\in \mathcal{P}_\mathrm{M}} \inf_{\psi\in \mathcal{L}}  J(\mathbb{Q},\psi)=
	\inf_{\psi\in \mathcal{L}}  \sup_{\mathbb{Q}\in \mathcal{P}_\mathrm{M}} J(\mathbb{Q},\psi)=
	J(\mathbb{Q}^\star_{\mathcal{N}},\psi^\star)
\label{eq24}
\end{equation}

The first equality is already established in \eqref{eq18}. Next, since $\psi^\star$ is an affine function, it follows that 
\begin{equation}
	J(\mathbb{Q}^\star_{\mathcal{N}},\psi^\star) = J(\boldsymbol{c},\boldsymbol{S}^\star;\boldsymbol{A}^\star,\boldsymbol{b}^\star)=
	\sup_{\mathbb{Q}\in \mathcal{P}_\mathrm{M}} \inf_{\psi\in \mathcal{L}}  J(\mathbb{Q},\psi)
\label{eq25}
\end{equation}
where the first equality is due to the definition of $J$ and the second equality is due to \eqref{eq3} and \eqref{eq18}.

\end{proof}

\section{Extension to Dynamic Systems: The MC-MDRKF Algorithm}
\subsection{Signal model}
The state equation and the measurement equation of a linear dynamic system are defined as
\begin{equation}
	\left.\begin{array}{l}
\boldsymbol{x}_{t}=\boldsymbol{F}_{t} \boldsymbol{x}_{t-1}+\boldsymbol{G}_{t} \boldsymbol{w}_{t} \\
\boldsymbol{y}_{t}^i=\boldsymbol{H}_{t}^i \boldsymbol{x}_{t}+\boldsymbol{v}_{t}^i
\end{array}\right\} \quad \forall t \in \mathbb{N},i=1,\dots,p
\label{eq26}
\end{equation}
where $i$ is the sensor index, $t$ is the time index, $\boldsymbol{x}_t \in \mathbb{R}^n$ is the state vector, $\boldsymbol{y}^i_t \in \mathbb{R}^{m_i}$ is the measurement vector of sensor $i$, $\boldsymbol{w}_t \in \mathbb{R}^r$ is the process noise, $\boldsymbol{v}^i_t \in \mathbb{R}^q$ is the measurement noise of sensor $i$, and $\boldsymbol{F}_t \in \mathbb{R}^{n \times n}$, $\boldsymbol{G}_t \in \mathbb{R}^{n \times r}$, $\boldsymbol{H}^i_t \in \mathbb{R}^{m_i \times n}$ are nominal system matrices.

The noise terms and the initial state are assumed Gaussian:
\begin{equation}
	\begin{array}{l}
\mathbb{E}\left\{\boldsymbol{w}_{t}\right\}=0, \mathbb{E}\left\{\boldsymbol{w}_{t} \boldsymbol{w}_{k}^{\top}\right\}=\boldsymbol{Q}_t \delta_{t k} \\
\mathbb{E}\left\{\boldsymbol{v}_{t}^{j}\right\}=0, \mathbb{E}\left\{\boldsymbol{v}_{t}^{j}\left(\boldsymbol{v}_{k}^{j}\right)^{\top}\right\}=\boldsymbol{R}^{j}_t \delta_{t k} \\
\mathbb{E}\left\{\boldsymbol{x}_{0}\right\}=\hat{\boldsymbol{x}}_{0}, \operatorname{Var}(\boldsymbol{x}_0) =\boldsymbol{V}_{0} \\
\left\{\boldsymbol{w}_{t}\right\},\left\{\boldsymbol{v}_{t}^{i}\right\}, \text { and }\left\{\boldsymbol{x}_{0}\right\} \text { are mutually independent }
\end{array}
\label{eq27}
\end{equation}

As an extension of the static case, it is assumed the correlations between sensor noises $\boldsymbol{v}_{t}^{i} (i=1,...,p)$ are unknown. Additionally, the exact values of system parameters such as $\boldsymbol{F}_t$, $\boldsymbol{G}_t$, $\boldsymbol{H}_t^i$, $\boldsymbol{Q}_t$, and $\boldsymbol{R}_t^i$ may be uncertain. Consequently, the true distribution $\mathbb{Q}$ of $\boldsymbol{z}_t=[\boldsymbol{x}_t^{\top},\boldsymbol{y}_t^{\top}]^{\top}$, departing from the nominal distribution $\mathbb{P}$, is unknown, making the estimation problem ill-defined. To address this, the conditional mean $\hat{\boldsymbol{x}}_t$ and covariance matrix $\boldsymbol{V}_t$ of $\boldsymbol{x}_t$ given the observation history $\boldsymbol{Y}_t$ are estimated under a worst-case distribution $\mathbb{Q}$, constructed recursively. 

\subsection{Solution}
The iterative prediction-correction estimation of $\boldsymbol{x}_t(t=1,2, ...) $ is as follows, given the marginal distribution $\mathbb{Q}_{\boldsymbol{x}_0} = \mathcal{N}_n(\hat{\boldsymbol{x}}_0, \boldsymbol{V}_0)$ and the conditional distribution $\mathbb{Q}_{\boldsymbol{x}_{t-1}|\boldsymbol{Y}_{t-1}} = \mathcal{N}_n(\hat{\boldsymbol{x}}_{t-1}, \boldsymbol{V}_{t-1})$.

The prediction step is conducted in the fusion center by combining each sensor's previous state estimate $\mathbb{Q}_{\boldsymbol{x}_{t-1}|\boldsymbol{Y}_{t-1}}$ with its nominal transition kernel $\mathbb{P}_{\boldsymbol{x}_t,\boldsymbol{y}_t^i|\boldsymbol{x}_{t-1}}$ to generate a series of pseudo-nominal distribution $\mathbb{P}_{\boldsymbol{x}_t,\boldsymbol{y}_t^i|\boldsymbol{Y}_{t-1}}$, which is defined as
\begin{equation}
\begin{aligned}
	&\mathbb{P}_{\boldsymbol{x}_t,\boldsymbol{y}_t^i \mid \boldsymbol{Y}_{t-1}}\left(B \mid \boldsymbol{Y}_{t-1}\right)\\
	& =\int_{\mathbb{R}^{n}} \mathbb{P}_{\boldsymbol{x}_t,\boldsymbol{y}_t^i \mid \boldsymbol{x}_{t-1}}\left(B \mid \boldsymbol{x}_{t-1}\right) \mathbb{Q}_{\boldsymbol{x}_{t-1} \mid \boldsymbol{Y}_{t-1}}\left(\mathrm{d} \boldsymbol{x}_{t-1} \mid \boldsymbol{Y}_{t-1}\right)
\end{aligned}
\label{eq29}
\end{equation}
for every Borel set $B \subseteq \mathbb{R}^{n+m}$ and observation history $\boldsymbol{Y}_{t-1} \in \mathbb{R}^{m \times (t-1)}$.
By the formula for the convolution of two multivariate Gaussians, we have 
\begin{equation}
\mathbb{P}_{\boldsymbol{x}_t,\boldsymbol{y}_t^i \mid \boldsymbol{Y}_{t-1}} = \mathcal{N}_{n+m_i}(\boldsymbol{\mu}_t^i,\boldsymbol{\Sigma}_t^i),i=1,\dots,p
\label{eq30_}
\end{equation}
where
\begin{equation}
\boldsymbol{\mu}_{t}^i=\left[\begin{array}{c}
\boldsymbol{\mu}_{x, t} \\
\boldsymbol{\mu}_{y^i, t}
\end{array}\right]=\left[\begin{array}{c}
\boldsymbol{F}_{t-1} \\
\boldsymbol{H}_{t}^i \boldsymbol{F}_{t-1}
\end{array}\right] \hat{\boldsymbol{x}}_{t-1 \mid t-1}
\label{eq30}
\end{equation}
and
\begin{equation}
	\begin{array}{l}
\boldsymbol{\Sigma}_{t}^i=\left[\begin{array}{c}
\boldsymbol{F}_{t-1} \\
\boldsymbol{H}_{t}^i \boldsymbol{F}_{t-1}
\end{array}\right] \boldsymbol{V}_{t-1}\left[\begin{array}{c}
\boldsymbol{F}_{t-1} \\
\boldsymbol{H}_{t}^i \boldsymbol{F}_{t-1}
\end{array}\right]^{T}+ \\
{\left[\begin{array}{cc}
\boldsymbol{G}_{t-1} \boldsymbol{Q}_{t-1}^{\frac{1}{2}} & \mathbf{0} \\
\boldsymbol{H}_{t}^i \boldsymbol{G}_{t-1} \boldsymbol{Q}_{t-1}^{\frac{1}{2}} & (\boldsymbol{R}_{t}^i)^{\frac{1}{2}}
\end{array}\right]\left[\begin{array}{cc}
\boldsymbol{G}_{t-1} \boldsymbol{Q}_{t-1}^{\frac{1}{2}} & \mathbf{0} \\
\boldsymbol{H}_{t}^i \boldsymbol{G}_{t-1} \boldsymbol{Q}_{t-1}^{\frac{1}{2}} & (\boldsymbol{R}_{t}^i)^{\frac{1}{2}}
\end{array}\right]^{T}}
\end{array}
\label{eq31}
\end{equation}

In the update step, the goal is to find a joint a priori distribution $\mathbb{Q}_{\boldsymbol{x}_t,\boldsymbol{y}_t\mid \boldsymbol{Y}_{t-1}}$ that robustified against marginal distributional uncertainty by solving \eqref{original minimax}. A refined a posteriori estimate $\mathbb{Q}_{\boldsymbol{x}_t \mid \boldsymbol{Y}_t}$, which is the solution of $\psi$ in \eqref{original minimax}, is then obtained:
\begin{equation}
\left\{
\begin{aligned}
\hat{\boldsymbol{x}}_{t\mid t} &= \boldsymbol{\mu}_{x,t} + \boldsymbol{S}_{xy,t}^\star (\boldsymbol{S}_{yy,t}^\star)^{-1}(\boldsymbol{y}-\boldsymbol{\mu}_{y,t})  \\
\boldsymbol{V}_{t\mid t} &= \boldsymbol{S}_{xx,t}^\star - \boldsymbol{S}_{xy,t}^\star(\boldsymbol{S}_{yy,t}^\star)^{-1}\boldsymbol{S}_{yx,t}^\star
\end{aligned}
\right.
\label{sov}
\end{equation}

Algorithm \ref{alg1} summarizes the proposed MC-MDRKF. Compared to \cite{niu2023marginal}, our method seeks the minimax solution over a new marginal distributional uncertainty set defined by moment constraints, thus achieving better robustness, as demonstrated in Section V. Note when $\gamma_{2,i},\gamma_{3,i}=0$, the marginal distributional uncertainty declines and the MC-MDRKF yields the optimal estimation as the canonical centralized Kalman filter.

\begin{algorithm}
\caption{Marginal Distributionally Robust Kalman Filter}
\label{alg1}
\small
\begin{algorithmic}
\REQUIRE Estimate at time $t-1$, $\hat{\boldsymbol{x}}_{t-1|t-1}$, and covariance $\boldsymbol{V}_{t-1|t-1}$
\STATE \textbf{Prediction step:}
\FOR{$i = 1$ \TO $p$}
  \STATE Compute $(\boldsymbol{\mu}_{t}^i,\boldsymbol{\Sigma}_{t}^i)$ with \eqref{eq30} and \eqref{eq31}
\ENDFOR
\STATE \textbf{Each node sends its measurement $\boldsymbol{y}_t^i$ to the fusion center}
\STATE \textbf{The fusion center formulates the marginal uncertainty set $\mathcal{P}$ with \eqref{def3}}
\STATE \textbf{Update step:}
\STATE Solve problem \eqref{maxmin-II} and obtain the estimator with \eqref{sov}
\ENSURE Estimate at time $t$, $\hat{\boldsymbol{x}}_{t|t}$, and covariance $\boldsymbol{V}_{t|t}$
\end{algorithmic}
\end{algorithm}

\section{Experiment}
The proposed algorithm is tested on a commonly used multi-sensor target tracking scenario under marginal distributional uncertainty, as in \cite{sun2004multi, yan2013optimal, lin2019globally, tian2016multi, feng2012optimal}. The real system dynamics are described by the following state-space model:
\begin{equation}
	\begin{array}{l}
		\boldsymbol{x}_{t+1} = \left[\begin{array}{ccc}
			1 & T_{s} & T_{s}^{2} / 2 \\
			0 & 1 & T_{s} \\
			0 & 0 & 1
		\end{array}\right] \boldsymbol{x}_{t} + \boldsymbol{G}_{t} w_{t} \\
		y^{i}_{t} = \boldsymbol{H}_{i} \boldsymbol{x}^{i}_{t} + v^{i}_{t}, \quad i=1,2,3 \\
		v^{i}_{t} = \beta_{i} w_{t-1} + \eta^{i}_{t}
	\end{array}
	\label{eq32}
\end{equation}
where $T_s=0.1$ is the sampling period. The state vector $\boldsymbol{x}_t = [s_t; \dot{s}_t; \ddot{s}_t]$ consists of the object's position, velocity, and acceleration at time $tT_s$, respectively. The experiment involves 600 time sampling points and a 1000-run Monte Carlo simulation, with the initial estimate being $\hat{\boldsymbol{x}}_0 = \boldsymbol{0}$ and $\boldsymbol{V}_0 = 100\boldsymbol{I}_3$. The rest of the settings align with \cite{sun2004multi}. 

The nominal system model matches the real one except for $v^{i}_{t} = \eta^{i}_{t}$. Candidate methods include the proposed MC-MDRKF, KF on sensor 1, covariance intersection (CI)\cite{chen2002estimation} and the MDRKF with KL uncertainty set extended from \cite{niu2023marginal}. 


The uncertainty set parameters $\gamma_{i}, c_i$ have both been optimally tuned. Figure \ref{fig2} displays the MSE of $s_t$ versus time.

\begin{figure}[htbp]
    \centering
    \includegraphics[width=0.8\linewidth]{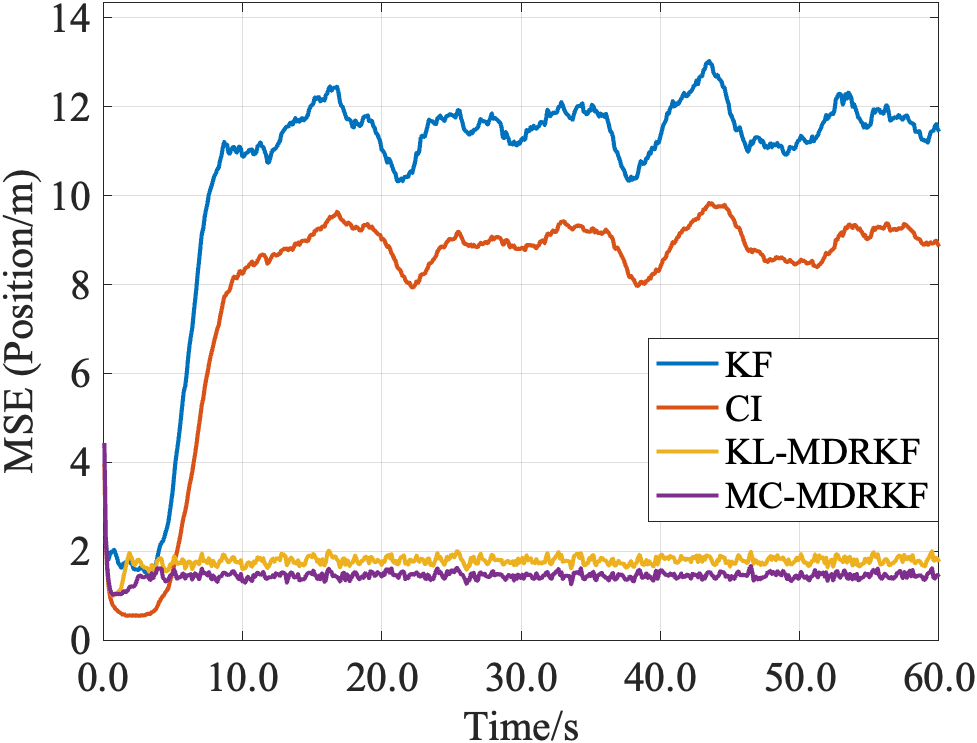}
        \caption{Comparison of MSE of target position estimation}
    \label{fig2}
\end{figure}

\begin{table}[ht]
    \centering
    \begin{threeparttable}
        
        \begin{tabular}{l|cccc}
            \toprule
            MSE & KF & CI & KL-MDRKF & \textbf{MC-MDRKF} \\
            \midrule
            $x$ (m) & 10.5345 & 7.9864 &  1.7764 & \textbf{1.4468} \\
            $v$ (m/s) & 30.2828 & 23.0917 & 8.5224 & \textbf{1.8426} \\
            $a$ (m/s$^2$) & 20.3734 & 16.7782 & 13.5665 & \textbf{2.6792} \\
            \bottomrule
        \end{tabular}
        \caption{MSE of different estimation methods}
        \label{tab:mse_comparison}
    \end{threeparttable}
\end{table}

In this multi-sensor tracking scenario with marginal distributional uncertainty, the MC-MDRKF significantly outperforms the classical CI method, KF, and KL-MDRKF, consistently achieving the lowest MSE across position, velocity, and acceleration estimates. This superior performance is due to its enhanced ability to depict marginal uncertainties and the corresponding minimax optimization. These results highlight the MC-MDRKF’s effectiveness and robustness in dynamic multi-sensor systems, marking a significant advancement in state estimation under uncertainty and providing a solid foundation for future research in robust multi-sensor fusion methods.

%
%

\section{Conclusions}

This paper proposes a novel MC-MDRKF algorithm based on a moment-constrained marginal distributional uncertainty set for robust multi-sensor state estimation in case of unknown noise correlation. By formulating and solving a corresponding minimax optimization problem, the proposed method enhances robustness against marginal uncertainty and achieves significant advantages over traditional KL divergence-based methods. We proved that the problem can be reformulated as a convex optimization problem, making it efficiently solvable. Experimental result validates the superiority of the MC-MDRKF, providing valuable insights and a solid foundation for future research in multi-sensor information fusion and robust state estimation.

%

\bibliographystyle{IEEEtran}
\bibliography{mybibfile}

\begin{thebibliography}{10}
\providecommand{\url}[1]{#1}
\csname url@samestyle\endcsname
\providecommand{\newblock}{\relax}
\providecommand{\bibinfo}[2]{#2}
\providecommand{\BIBentrySTDinterwordspacing}{\spaceskip=0pt\relax}
\providecommand{\BIBentryALTinterwordstretchfactor}{4}
\providecommand{\BIBentryALTinterwordspacing}{\spaceskip=\fontdimen2\font plus
\BIBentryALTinterwordstretchfactor\fontdimen3\font minus
  \fontdimen4\font\relax}
\providecommand{\BIBforeignlanguage}[2]{{%
\expandafter\ifx\csname l@#1\endcsname\relax
\typeout{** WARNING: IEEEtran.bst: No hyphenation pattern has been}%
\typeout{** loaded for the language `#1'. Using the pattern for}%
\typeout{** the default language instead.}%
\else
\language=\csname l@#1\endcsname
\fi
#2}}
\providecommand{\BIBdecl}{\relax}
\BIBdecl

\bibitem{8861414}
A.~K. Gostar, T.~Rathnayake, R.~Tennakoon, A.~Bab-Hadiashar, G.~Battistelli,
  L.~Chisci, and R.~Hoseinnezhad, ``Centralized cooperative sensor fusion for
  dynamic sensor network with limited field-of-view via labeled multi-bernoulli
  filter,'' \emph{IEEE Transactions on Signal Processing}, vol.~69, pp.
  878--891, 2021.

\bibitem{8805459}
Y.~Wang, Y.~Sun, and V.~Dinavahi, ``Robust forecasting-aided state estimation
  for power system against uncertainties,'' \emph{IEEE Transactions on Power
  Systems}, vol.~35, no.~1, pp. 691--702, 2020.

\bibitem{9479686}
D.~E. Clark, ``Multi-sensor network information for linear-{Gaussian}
  multi-target tracking systems,'' \emph{IEEE Transactions on Signal
  Processing}, vol.~69, pp. 4312--4325, 2021.

\bibitem{willner1976kalman}
D.~Willner, C.-B. Chang, and K.-P. Dunn, ``Kalman filter algorithms for a
  multi-sensor system,'' in \emph{1976 IEEE conference on decision and control
  including the 15th symposium on adaptive processes}.\hskip 1em plus 0.5em
  minus 0.4em\relax IEEE, 1976, pp. 570--574.

\bibitem{liu2017robust}
W.-Q. Liu, X.-M. Wang, and Z.-L. Deng, ``Robust centralized and weighted
  measurement fusion {Kalman} estimators for uncertain multisensor systems with
  linearly correlated white noises,'' \emph{Information Fusion}, vol.~35, pp.
  11--25, 2017.

\bibitem{chen2015networked}
B.~Chen, W.~Zhang, G.~Hu, and L.~Yu, ``Networked fusion kalman filtering with
  multiple uncertainties,'' \emph{IEEE transactions on Aerospace and Electronic
  Systems}, vol.~51, no.~3, pp. 2232--2249, 2015.

\bibitem{fan2024distributionally}
Z.~Fan, R.~Ji, and M.~A. Lejeune, ``Distributionally robust portfolio
  optimization under marginal and copula ambiguity,'' \emph{Journal of
  Optimization Theory and Applications}, vol. 203, no.~3, pp. 2870--2907, 2024.

\bibitem{niu2023marginal}
D.~Niu, E.~Song, Z.~Li, L.~Zhang, T.~Ma, J.~Gu, and Q.~Shi, ``A marginal
  distributionally robust {MMSE} estimation for a multisensor system with
  kullback-leibler divergence constraints,'' \emph{IEEE Transactions on Signal
  Processing}, vol.~71, pp. 3772--3787, 2023.

\bibitem{zorzi2016robust}
M.~Zorzi, ``Robust {Kalman} filtering under model perturbations,'' \emph{IEEE
  Transactions on Automatic Control}, vol.~62, no.~6, pp. 2902--2907, 2016.

\bibitem{zorzi2019distributed}
------, ``Distributed {Kalman} filtering under model uncertainty,'' \emph{IEEE
  Transactions on Control of Network Systems}, vol.~7, no.~2, pp. 990--1001,
  2019.

\bibitem{shafieezadeh2018wasserstein}
S.~Shafieezadeh~Abadeh, V.~A. Nguyen, D.~Kuhn, and P.~M. Mohajerin~Esfahani,
  ``Wasserstein distributionally robust {Kalman} filtering,'' \emph{Advances in
  Neural Information Processing Systems}, vol.~31, 2018.

\bibitem{Rahimian2019DistributionallyRO}
H.~Rahimian and S.~Mehrotra, ``Distributionally robust optimization: A
  review,'' \emph{arXiv preprint arXiv:1908.05659}, 2019.

\bibitem{wang2021robust}
S.~Wang, Z.~Wu, and A.~Lim, ``Robust state estimation for linear systems under
  distributional uncertainty,'' \emph{IEEE Transactions on Signal Processing},
  vol.~69, pp. 5963--5978, 2021.

\bibitem{wang2021distributionally}
S.~Wang and Z.-S. Ye, ``Distributionally robust state estimation for linear
  systems subject to uncertainty and outlier,'' \emph{IEEE Transactions on
  Signal Processing}, vol.~70, pp. 452--467, 2021.

\bibitem{fan1953minimax}
K.~Fan, ``Minimax theorems,'' \emph{Proceedings of the National Academy of
  Sciences}, vol.~39, no.~1, pp. 42--47, 1953.

\bibitem{grant2014cvx}
M.~Grant and S.~Boyd, ``{CVX}: Matlab software for disciplined convex
  programming, version 2.1,'' 2014.

\bibitem{boyd2004convex}
S.~Boyd and L.~Vandenberghe, \emph{Convex Optimization}.\hskip 1em plus 0.5em
  minus 0.4em\relax Cambridge, U.K.: Cambridge University Press, 2004.

\bibitem{sun2004multi}
S.-L. Sun and Z.-L. Deng, ``Multi-sensor optimal information fusion {Kalman}
  filter,'' \emph{Automatica}, vol.~40, no.~6, pp. 1017--1023, 2004.

\bibitem{yan2013optimal}
L.~Yan, X.~R. Li, Y.~Xia, and M.~Fu, ``Optimal sequential and distributed
  fusion for state estimation in cross-correlated noise,'' \emph{Automatica},
  vol.~49, no.~12, pp. 3607--3612, 2013.

\bibitem{lin2019globally}
H.~Lin and S.~Sun, ``Globally optimal sequential and distributed fusion state
  estimation for multi-sensor systems with cross-correlated noises,''
  \emph{Automatica}, vol. 101, pp. 128--137, 2019.

\bibitem{tian2016multi}
T.~Tian, S.~Sun, and N.~Li, ``Multi-sensor information fusion estimators for
  stochastic uncertain systems with correlated noises,'' \emph{Information
  Fusion}, vol.~27, pp. 126--137, 2016.

\bibitem{feng2012optimal}
J.~Feng and M.~Zeng, ``Optimal distributed {Kalman} filtering fusion for a
  linear dynamic system with cross-correlated noises,'' \emph{International
  Journal of Systems Science}, vol.~43, no.~2, pp. 385--398, 2012.

\bibitem{chen2002estimation}
L.~Chen, P.~O. Arambel, and R.~K. Mehra, ``Estimation under unknown
  correlation: Covariance intersection revisited,'' \emph{IEEE Transactions on
  Automatic Control}, vol.~47, no.~11, pp. 1879--1882, 2002.

\end{thebibliography}

\end{document}